\documentclass[12pt]{amsart}
\usepackage{amsmath,amssymb,amsbsy,amsfonts,amsthm,latexsym,
                        amsopn,amstext,amsxtra,euscript,amscd,color}
\usepackage[english]{babel}
\def\F{{\mathbb F}}
\def\N{{\mathbb N}}
\def\Z{{\mathbb Z}}

\newtheorem{theorem}{Theorem}
\newtheorem{lemma}[theorem]{Lemma}

\newcommand{\quash}[1]{}

\def\\{\cr}
\def\({\left(}
\def\){\right)}
\def\[{\left[}
\def\]{\right]}
\def\<{\langle}
\def\>{\rangle}
\def\fl#1{\left\lfloor#1\right\rfloor}
\def\rf#1{\left\lceil#1\right\rceil}
\def\le{\leqslant}
\def\ge{\geqslant}

\def\mand{\qquad\mbox{and}\qquad}

\def\cI{{\mathcal I}}
\def\cJ{{\mathcal J}}

\def\cM{{\mathcal M}}

\def\cS{{\mathcal S}}

\def\cX{{\mathcal X}}

\def\cZ{{\mathcal Z}}

\def \thetaord{\vartheta_{4}\,}

\begin{document}

\title{Covering sets for limited-magnitude errors}
\author[Z. Chen] {Zhixiong Chen}
\address{School of Applied Mathematics, Putian University, \\ Putian, Fujian
351100, P.R. China}
\email{ptczx@126.com}

\author[I. E. Shparlinski] {Igor E. Shparlinski}
\address{Department of Pure Mathematics, University of New South Wales,
Sydney, NSW 2052, Australia}
\email{igor.shparlinski@unsw.edu.au}

\author[A.  Winterhof]{Arne Winterhof}
\address{Johann Radon Institute for Computational and Applied Mathematics,
Austrian Academy of Sciences,
Altenberger Stra\ss e 69, A-4040 Linz,
Austria}
\email{arne.winterhof@oeaw.ac.at}

\begin{abstract}
 For  a set
 $\cM=\{ -\mu,-\mu+1,\ldots, \lambda\}\setminus\{0\}$
with  non-negative integers $\lambda,\mu<q$ not both $0$, a subset $\cS$ of the residue class ring  $\Z_q$
modulo an integer $q\ge 1$ is called a $(\lambda,\mu;q)$-\emph{covering set} if
$$
\cM \cS=\{ms \bmod q~:~m\in \cM,\ s\in \cS\}=\Z_q.
$$
Small covering sets  play an important role in codes correcting limited-magnitude
errors. We give an explicit construction of a $(\lambda,\mu;q)$-covering set
$\cS$ which is of the size $q^{1 + o(1)}\max\{\lambda,\mu\}^{-1/2}$ for 
almost all integers $q\ge 1$ and of optimal size  $p\max\{\lambda,\mu\}^{-1}$
if $q=p$ is prime. 
Furthermore, using a bound on the fourth moment of character sums of Cochrane and Shi
we prove the bound
$$\omega_{\lambda,\mu}(q)\le q^{1+o(1)}\max\{\lambda,\mu\}^{-1/2},$$
for any integer $q\ge 1$, however the proof of this bound is not constructive. 
\end{abstract}



\maketitle

\section{Introduction}

Codes correcting limited magnitude  errors have
been introduced in flash memory devices, which are widely used
nowadays, see for example the recent survey~\cite{GYD}. 

Let $q$ be a positive integer and $\cM=\{ -\mu,-\mu+1,\ldots, \lambda\}\setminus\{0\}$ for 
non-negative integers $\lambda,\mu<q$ with $\lambda+\mu>0$. Following  Kl{\o}ve and Schwartz~\cite{KS}, for  a
set $\cS\subseteq \Z_q$ in the residue class ring $\Z_q$ modulo an integer $q$,  we consider the
product set
$$
\cM \cS=\{ms \bmod q~:~m\in \cM,\ s\in \cS\}
$$
and define
$$
\nu_{\lambda,\mu}(q,r)=\max_{\cS\subseteq \Z_q}\{\#(\cM \cS)~:~\#
\cS=r\},\quad r\ge 1,
$$
where $\# \cZ$ denotes  the cardinality of a set $\cZ$.
We have the trivial bound
$$
\max\{r, \lambda+\mu \}\le \nu_{\lambda,\mu}(q,r)\le \min \{(\lambda+\mu)r, q\}.
$$

Very recently, Kl{\o}ve and Schwartz~\cite{KS} have introduced the
notion of $(\lambda,\mu;q)$-covering sets. Namely, a subset
$\cS\subseteq \Z_q$ is called a \emph{$(\lambda,\mu;q)$-covering
set} of $\Z_q$ if $\cM \cS=\Z_q$. As stated in~\cite{KS}, the
problem of covering for certain parameters has applications such as
\emph{rewriting schemes}, see also~\cite{JLSB}. The task is to find
$(\lambda,\mu;q)$-covering sets of size as small as possible. Define
$$
\omega_{\lambda,\mu}(q)=\min\{r\in \N~:~\nu_{\lambda,\mu}(q,r)=q\}.
$$
Clearly, we have the lower bound
\begin{equation}\label{eq: lower}
\omega_{\lambda,\mu}(q)\ge
\rf{\frac{q}{\lambda+\mu}},\quad \lambda+\mu<q.
\end{equation}
We  prove the general upper bound
\begin{equation}\label{eq: q upper}
\omega_{\lambda,\mu}(q)=O\(\frac{q(\log q)^{r(q)}}{\max\{\lambda,\mu\}^{1/2}}\),
\end{equation}
where $r(q)$ is the number of prime divisors of $q$. 
In many cases our {\em constructive} method provides stronger bounds. In particular,
if $q=p$ is a prime we get
$$\omega_{\lambda,\mu}(p)\le 2 \rf{p/\max\{\lambda,\mu\}} -1,
$$
which is consistent with the lower bound~\eqref{eq: lower}.

Note that we can always assume that 
$\lambda+\mu< q-1$ as otherwise $\cS=\{0,1\}$ is trivially a $(\lambda,\mu;q)$-covering set of smallest possible cardinality. 

Although $r(q)$ is typically quite small, for some $q$ the bound~\eqref{eq: q upper} can be trivial. 
However, using a bound on the fourth moment of character sums of Cochrane and Shi~\cite{CochShi} we prove the general bound
$$\omega_{\lambda,\mu}(q)\le q^{1+o(1)}\max\{\lambda,\mu\}^{-1/2},$$
however, the proof is not constructive. 

We also consider some questions which appear in the case of very small 
values of $\lambda$ and $\mu$. For instance,  Kl{\o}ve and
Schwartz~\cite[Corollary~3]{KS} have given a description of the 
integers $q$ which admit an explicit formula for $\omega_{2,1}(q)$.
This description involves the property of the multiplicative order 
of $2$ modulo all prime divisors of $q$. We show that classical 
number theoretic tools allow to obtain an asymptotic formula for the
number of such integers $q \le Q$ (this question has been investigated 
numerically in~\cite{KS}).

Finally, we discuss the approach of~\cite{KS} to estimating 
$\omega_{\lambda,\mu}(p)$ (for a prime $p$) via the number of residues
 of a sequence of 
consecutive powers of 
a given primitive root modulo $p$ in a short   interval. 
We show that several recently obtained results due to 
Bourgain~\cite{Bour4,Bour5} indicate that this 
approach has no chance to succeed. 

Throughout this work, the implied constants in the symbols `$O$',
and `$\ll$' are absolute.
We recall that the notations $U = O(V)$
and $U \ll V$ are both equivalent to the assertion that the
inequality $|U|\le cV$ holds for some constant $c>0$.

As usual, for an integer $q$, we use $\varphi(q)$ to denote the Euler function 
of $q$ and $\tau(q)$ the number of integer positive divisors 
of $q$. 

Since $\omega$ has already got another meaning in this 
work, which stems from notation of~\cite{KS}, then, 
as before, we use $r(q)$ 
for the number of distinct prime divisors of $q$.  

We also use $\Z_q^*$ to denote the set of invertible elements of $\Z_q$.

The letter $p$, with or without subscripts, always denotes a prime number.

\section{Construction}

Here we give explicit constructions of $(\lambda,\mu;q)$-covering sets

\begin{theorem}
\label{thm:Constr q} 
For any integer  $q \ge 1$ and non-negative integers $\lambda,\mu$ 
with $\lambda+\mu\ge 1$ in time $\#\cS (\log q)^{O(1)}$ 
one can construct a $(\lambda,\mu;q)$-covering set $\cS \subseteq \Z_q$ with 
$$\#\cS  =O\(\frac{q(\log q)^{r(q)}}{\max\{\lambda,\mu\}^{1/2}}\).
$$
\end{theorem}

\begin{proof}
Note that $\cS$ is a $(\lambda,\mu;q)$-covering set whenever so is
$-\cS=\{q-s~:~s\in \cS\}$. Hence, we may restrict ourselves to the 
case $\lambda\ge \mu$ and note that $\{1,2,\ldots,\lambda\}\subseteq \cM$.

First we consider the case that
$q=p^\ell$ with a prime $p$.

If $\lambda<p$, we put $H=\rf{p/\lambda}-1$ and
\begin{equation*}
\begin{split}
 \cS=\{s_0+s_1&p~:~\\
  &s_0\in \{\pm j^{-1} \bmod p~:~1\le j \le H\}\cup \{0\}, 0\le s_1<p^{\ell-1}\}.
\end{split}
\end{equation*}
Clearly 
\begin{equation}\label{eq: prime}
\#\cS =(2H+1)p^{\ell-1}< 3p^\ell/\lambda.
\end{equation}

Let $a\equiv a_0+a_1p\bmod p^\ell$ be any integer with $0\le a_0<p$ and $0\le a_1<p^{\ell-1}$.
We have to show that
$$
a_0\equiv ms_0\bmod p
$$
for some $m\in \{1,2,\ldots,\lambda\}$ and $s_0\in \{\pm j^{-1} \bmod p: 1\le j \le
H\}\cup \{0\}$. Then taking $s_0+s_1p\in \cS$ with 
$$s_1\equiv
\(\frac{a_0-ms_0}{p}+a_1\)m^{-1}\bmod p^{\ell-1},$$ 
we derive
$$
a\equiv a_0+a_1p\equiv m(s_0+s_1p)\bmod p^\ell \in \cM \cS.
$$
If $a_0=0$, we take $s_0=0$, $m=1$, and $s_1=a_1$.
If $a_0\ne 0$, there are at least two
elements $r_1a_0, r_2a_0$ with $0\le r_1, r_2 \le
H$ such that
$$
0< (r_1a_0 -r_2a_0 \bmod p)\le \frac{p}{H+1}\le \lambda
$$
by the pigeon-hole principle. We take $m=(r_1-r_2)a_0 \bmod p\in \cM$
and $s_0=(r_1-r_2)^{-1}\bmod p \in \{\pm j^{-1}\bmod p~:~1\le j\le H\}$, and get
$$
a_0\equiv ms_0\bmod p.
$$
Hence, $\cM \cS=\Z_{p^\ell}$.

If $p^j\le \lambda <p^{j+1}$ for some $1\le j<\ell$, we take
$$
\cS=\{s_0+s_1p^j~:~s_0\in \{p^i~:~i=0,\ldots,j-1\}\cup\{0\},~0\le s_1<p^{\ell-j}\}.
$$
We show that any $a\equiv a_0+a_1p^j\bmod p^\ell$ with $0\le a_0<p^j$ and $0\le a_1<p^{\ell-j}$ can be written as
$a\equiv m s\bmod p^\ell$ with $1\le m\le \lambda$ and $s\in \cS$. 

If $a_0=0$, we take $m=1$ and $s=a\in \cS$. 
If $\gcd(a_0,p^j)=p^i$ for some $0\le i< j$, we take $m=a_0/p^i<p^j\le \lambda$ and $s=p^i+s_1p^j$ 
with 
$$s_1\equiv \(\frac{a_0}{p^i}\)^{-1}a_1 \bmod p^{\ell-j}.$$
Hence, we have $\cM\cS=\Z_{p^\ell}$ and 
\begin{equation}
\label{eq: power}
\#\cS=(j+1)p^{\ell-j}<(j+1)p^\ell/\lambda^{j/(j+1)}\le (j+1)p^\ell/\lambda^{1/2}.
\end{equation}

Now we assume that $q=p_1^{\ell_1}\cdots p_r^{\ell_r}$ is the prime decomposition of $q$ with different primes $p_1,\ldots,p_r$ and 
$p_1^{\ell_1}>p_2^{\ell_2}>\ldots >p_r^{\ell_r}$.
We  inductively construct a covering set $\cS \subseteq \Z_q$ of size
$$\# \cS = O\(q\lambda^{-1/2}(\log q)^r\).$$
For $r=1$ this result follows from~\eqref{eq: prime} and~\eqref{eq: power} 
and we assume $r\ge 2$.
We put $\widetilde q=p_1^{\ell_1}\cdots p_{r-1}^{\ell_{r-1}}$ and $p^\ell=p_r^{\ell_r}$. Note that $\gcd(\widetilde q,p)=1$ and $\widetilde q> q^{1/2}>p^\ell$.

If $\lambda <\widetilde q$, let $\cS_0$ be a $(\lambda,0;\widetilde q)$-covering set of size 
$$\#\cS_0=O\( \widetilde q\lambda^{-1/2}(\log \widetilde q)^{r-1}\)$$ 
which exists by induction.
Now we put 
$$\cS_1=\{ p^is_0~:~i=0,\ldots,\fl{\log \widetilde q/\log p},~s_0\in \cS_0\}.$$
Let $m\in \{1,\ldots,\lambda\}$, $s_0\in \cS_0$, be a solution of $ms_0\equiv a_0\bmod \widetilde q$ and $p^i$ be the largest power of $p$ which divides $m$, that is, 
$i\le \fl{\log \widetilde q/\log p}$. 
Then $m_1=m/p^i$ and $s_1=p^is_0\in \cS_1$ is another solution of
$$m_1s_1\equiv a_0\bmod \widetilde q\quad \mbox{ with }\gcd(m_1,p)=1.$$ 
Put 
$$s_2=m_1^{-1} \(a_1+\frac{a_0-m_1s_1}{\widetilde q}\)\bmod p^\ell$$ 
and verify that
$$m_1(s_1+s_2\widetilde q)\equiv a_0+a_1\widetilde q\bmod q.$$
Consequently, $\cS=\{0\le s<q~:~s \bmod p^\ell \in \cS_1\}$ is a $(\lambda,0;q)$-covering set of size 
$$\#\cS =O\( \widetilde q\lambda^{-1/2} (\log \widetilde q)^{r-1} \cdot p^\ell \log(\widetilde q)\)=O\(q\lambda^{-1/2}(\log q)^r\).$$

If $\widetilde q\le \lambda <q$, we write $a\equiv a_0+a_1\widetilde q \bmod q$ with $0\le a_0<\widetilde q$ and $0\le a_1<p^\ell$. 
If $a_0=0$, we take $m=1$ and $s=a$. Otherwise let $p^i$ be the largest power of $p$ dividing $a_0$
and take $m=a_0/p^i<\widetilde q\le \lambda$ and $s=p^i+s_1\widetilde q$ with $s_1\equiv (a_0/p^i)^{-1}a_1 \bmod p^\ell$.
Hence, 
$$\cS=\{s_0+s_1\widetilde q~:~s_0\in \{p^i ~:~ i=0,\ldots,\fl{ \log \widetilde q/\log p}\}\cup\{0\},~0\le s_1<p^\ell\}$$ 
is a $(\lambda,0;q)$-covering set of size
\begin{eqnarray*}\#\cS&\le& p^\ell(\fl{\log \widetilde q/\log p}+2)\\ &=& O\(q\log \widetilde q/\widetilde q^{1/2} \cdot p^{\ell/2}/\lambda^{1/2}\)=O\(q\log q/\lambda^{1/2}\)
\end{eqnarray*}
and the result follows since any $(\lambda,0;q)$-covering set is a $(\lambda,\mu;q)$-covering set.

Clearly, the inductive construction works in polynomial time 
per every element of $\cS$, which yields the desired complexity bound. 
\end{proof}

Using that 
\begin{equation}
\label{eq:sum tau2}
\sum_{q \le Q} 2^{2r(q)} \le \sum_{q \le Q} \tau^2(q) = (1+ o(1))Q (\log Q)^3, 
\end{equation}
as $Q \to \infty$, 
see~\cite[Chapter~1, Theorem~5.4]{Prach},
we see that for any $\varepsilon > 0$ the inequality  
\begin{equation}
\label{eq:good q}
r(q) < \varepsilon \frac{\log Q}{\log \log Q}
\end{equation}
fails for at most
$$
Q \exp\(-(\varepsilon \log 4 + o(1)) \frac{\log Q}{\log \log Q}\)
\ll Q \exp\(-\varepsilon \frac{\log Q}{\log \log Q}\)
$$ 
positive integers $q \le Q$. Indeed to derive this from~\eqref{eq:sum tau2} 
we simply discard all 
$$
q \le Q \exp\(-\varepsilon \frac{\log Q}{\log \log Q}\)
$$ 
and note that for 
$$
Q \exp\(-\varepsilon \frac{\log Q}{\log \log Q}\)
 < q \le Q
$$ 
we have $\log q = (1+o(1)) \log Q$. 
For the remaining values of $q$, satisfying~\eqref{eq:good q},  
the size of the set $\cS$ of Theorem~\ref{thm:Constr p} is 
$$
\#\cS  =O\(\frac{q^{1 + \varepsilon}}{\max\{\lambda,\mu\}^{1/2}}\).
$$

We now note that if $q=p$  is a prime, then we  always have $\lambda<p$ 
so the bound~\eqref{eq: prime}
applies and we obtain the following stronger result:

\begin{theorem}
\label{thm:Constr p} 
For any prime $p$ and non-negative integers $\lambda,\mu$ 
with $\lambda+\mu\ge 1$ in time $\#\cS (\log p)^{O(1)}$ 
one can construct a $(\lambda,\mu;p)$-covering set $\cS \subseteq \Z_p$ with 
$$\#\cS \le 2\lceil p/\max\{\lambda,\mu\}\rceil -1.$$
\end{theorem}

\section{Upper Bound}

As we have mentioned, Theorem~\ref{thm:Constr q}
applies to the majority of positive integers $q$, however there is
a set of integers $q$ for which it gives only a trivial estimate.
We now use a different approach to give a non-constructive bound 
on $\omega_{\lambda,\mu}(q)$ which applies to any $q$.

We start with recalling the following well-known estimates
on the divisor and  Euler functions
\begin{equation}
\label{eq:phitau}
 \tau(q) = q^{o(1)}
 \mand  \varphi(q) =q^{1+o(1)},
\end{equation}
as $q \to \infty$,
see~\cite[Chapter~1, Theorems~5.1 and~5.2]{Prach}.

We also need the following well-known consequence of the sieve of Eratosthenes.

\begin{lemma}\label{lem:erat} For any integers $q, U\ge 1$,
$$
\sum_{\substack{u=1 \\ \gcd(u,\,q)=1}}^U 1 = \frac{\varphi(q)}{q}U 
+ O(2^{r(q)}).
$$
\end{lemma}

\begin{proof} Using the M\"obius function $\mu(d)$ over the divisors of 
$q$ to detect
the co-primality condition and interchanging the order of summation, we 
obtain the
Legendre formula
$$
\sum_{\substack{u=1 \\ \gcd(u,\,q)=1}}^U 1 = \sum_{d|q}\mu(d)
\fl{\frac{U}{d}} =
U\sum_{d|q}\frac{\mu(d)}{d}+O\(\sum_{d|q}|\mu(d)|\)
$$
from which the result follows immediately.
\end{proof}

Let $\cX$ be the set of all multiplicative characters $\chi$ modulo $q$ and 
let $\cX^*$ be the set of non-principal characters $\chi \ne \chi_0$. 
We now  recall the bound of
Cochrane and Shi~\cite{CochShi} on the fourth moment of
character sums, which we present in the 
following slightly less precise form, which follows from~\cite[Theorem~1]{CochShi}
and~\eqref{eq:phitau}.

\begin{lemma}
   \label{lem:4th Moment} For  arbitrary integers
$U\ge 1$,  and $V$, the bound
$$
\sum_{\chi \in \cX}
\left| \sum_{u = V+1}^{V+U} \chi(u)\right|^4
\le   q^{1 + o(1)} U^2
$$
holds.
\end{lemma}

We now derive an extension of the result 
of Garaev and
Garcia~\cite[Theorem~2]{GaGa}, which is our main technical tool.

\begin{lemma}
\label{lem:prod} Let $\varepsilon>0$ be fixed
and $q$ be a sufficiently large positive 
integer. For any intervals $\cI = [K+1,K+M]$ 
and $\cJ = [L+1,L+N]$ with $\cI, \cJ \subseteq [1, q-1]$, and 
$M,N \ge q^\varepsilon$, 
for all but at most $q^{2+o(1)}M^{-1}N^{-1}$ elements $a \in \Z_q^*$, the congruence  
\begin{equation}
\label{eq:cong}
a \equiv mn \bmod q, \qquad m \in \cI, \ n \in \cJ,
\end{equation}
has a solution. 
\end{lemma}

\begin{proof} Let $\cI^*$ and $\cJ^*$ denote the set 
of integers $m \in \cI$ and $n\in \cJ$, respectively with $\gcd(m,q)= 
\gcd(n,q) = 1$. 
Since $2^{r(q)} \le \tau(q)$, we conclude from~\eqref{eq:phitau} and Lemma~\ref{lem:erat}
that
\begin{equation}
\label{eq:IJ*}
\# \cI^*= (1+o(1))\frac{\varphi(q)}{q}M   \mand \# \cJ^*=(1+o(1))\frac{\varphi(q)}{q}N.
\end{equation}

Using the orthogonality of characters, we see that the number $J(a)$
of solutions to~\eqref{eq:cong} can be written as
$$
J(a)   =  \frac{1}{\varphi(q)} \sum_{m \in \cI}
\sum_{n\in \cJ} \sum_{\chi\in \cX} \chi(mna^{-1}).
$$
Changing the order of summation and separating the contribution 
$\# \cI^* \# \cJ^*/\varphi(q)$ of the principal character, 
we obtain
$$
J(a)  - \frac{\# \cI^* \# \cJ^*}{\varphi(q)}
= \frac{1}{\varphi(q)} \sum_{\chi\in \cX^*} \chi(a^{-1}) \sum_{m \in \cI} 
\chi(m)
\sum_{n\in \cJ} \chi(n).
$$
Hence
\begin{equation*}
\begin{split}
&\sum_{a\in \Z_q^*}  
\(J(a)  - \frac{\# \cI^* \# \cJ^*}{\varphi(q)}\)^2\\
& \quad = \frac{1}{\varphi^2(q)} 
\sum_{a\in \Z_q^*} 
\sum_{\chi_1,\chi_2 \in \cX^*} \chi_1(a^{-1}) \chi_2(a^{-1})
\sum_{m_1 \in \cI} \chi_1(m_1) \sum_{m_2 \in \cI}  \chi_2(m_2)\\
&\qquad \qquad \qquad \qquad\qquad \qquad \qquad \qquad \qquad \qquad 
\sum_{n_1\in \cJ} \chi_1(n_1) \sum_{n_2\in \cJ} \chi_2(n_2)\\
& \quad =  \frac{1}{\varphi^2(q)} 
\sum_{\chi_1,\chi_2 \in \cX^*} \sum_{m_1 \in \cI} \chi_1(m_1) \sum_{m_2 \in \cI}  \chi_2(m_2)
\sum_{n_1\in \cJ} \chi_1(n_1) \sum_{n_2\in \cJ} \chi_2(n_2) \\
&\qquad \qquad \qquad \qquad \qquad \qquad \qquad \qquad \qquad 
\sum_{a\in \Z_q^*} \chi_1(a^{-1}) \chi_2(a^{-1}).
\end{split}
\end{equation*}
By the orthogonality of characters again, we  see that the inner sum
vanishes, unless $\chi_1=\chi_2$ (in which case it is equal to $\varphi(q)$). 
Hence
\begin{equation*}
\begin{split}
\sum_{a\in \Z_q^*}    
\(J(a)  - \frac{\# \cI^* \# \cJ^*}{\varphi(q)}\)^2&\\
 =  \frac{1}{\varphi(q)} &
\sum_{\chi \in \cX^*} \(\sum_{m \in \cI} \chi(m) \)^2
\(\sum_{n\in \cJ} \chi(n)\)^2.
\end{split}
\end{equation*}
Thus, from the Cauchy-Schwarz inequality, Lemma~\ref{lem:4th Moment} and the 
bound~\eqref{eq:phitau}
we obtain 
$$
\sum_{a\in \Z_q^*}    
\(J(a)  - \frac{\# \cI^* \# \cJ^*}{\varphi(q)}\)^2 \le NM q^{o(1)}.
$$
Hence, using~\eqref{eq:phitau} and~\eqref{eq:IJ*}, we see that  $J(a)=0$ is possible for at most 
$$
NM q^{o(1)} \frac{\varphi(q)^2}
{(\# \cI^* \# \cJ^*)^2}= q^{2+o(1)}M^{-1}N^{-1}
$$
values of $a\in \Z_q^*$. 
\end{proof}

Now we are able to prove the main result of this section.

\begin{theorem}
\label{thm:bound}
  For any integer $q$ and any positive integers $\lambda,\mu < q$ with $\lambda+\mu\ge 1$ we have
  $$\omega_{\lambda,\mu}(q)  \le  \frac{q^{1+o(1)}}{\max\{\lambda,\mu\}^{1/2}}.$$
\end{theorem}

\begin{proof} We define $\omega_{\lambda,\mu}^*(q)$ in exactly the
same way as   $\omega_{\lambda,\mu}(q)$ with respect to $\Z_q^*$ 
instead of $\Z_q$. 
Collecting  the elements $a \in \Z_q$ with the same value $d = \gcd(a,q)$, 
we see that 
$$
\omega_{\lambda,\mu}(q)  \le \sum_{d \mid q} \omega_{\lambda,\mu}^*(q/d).
$$

We now see from~\eqref{eq:phitau} that it is enough to show that for an 
arbitrary parameter $\varepsilon > 0$, 
we have
\begin{equation}
\label{eq:omega*}
\omega_{\lambda,\mu}^*(q)  \le  \frac{q^{1+\varepsilon}}{\max\{\lambda,\mu\}^{1/2}}
\end{equation}
provided that $\lambda,\mu < q$.  

Without loss of generality we restrict ourselves to the case
$\lambda\ge \mu$ and choose 
$$
\Delta = \sqrt{\lambda},
$$

We can also assume that $\lambda \ge q^{\varepsilon}$
as otherwise the bound is trivial.
Hence 
\begin{equation}
\label{eq:Delta}
\lambda \ge \Delta \ge  q^{\varepsilon/2}.
\end{equation}

Set
$$
\cS_0=\{1,\ldots, \rf{\lambda^{-1} \Delta q^{1 + \varepsilon/2}}\}.
$$
Taking into account~\eqref{eq:Delta} 
we infer from Lemma~\ref{lem:prod} that all but
a set $\cS_1$ of
$$
\# \cS_1  \le q^{1 +o(1)} \Delta^{-1} 
$$
residue classes $a \in \Z_q^*$ can be represented as 
$ms\equiv a\bmod q$ with $1\le m\le \lambda$ and $s\in \cS_0$.

Setting $\cS=\cS_0\cup \cS_1$, after elementary  calculations,
we derive~\eqref{eq:omega*} and conclude the proof. 
\end{proof}

\section{Some Special Cases}

Kl{\o}ve and
Schwartz~\cite{KS} have also studied $\omega_{\lambda,\mu}(p)$ for primes $p$ and very small 
values of $\lambda + \mu$ and presented several explicit formulas.

 First we observe that the expression that appears
in the formula for $\omega_{2,0}(q)$ with odd $q$ is 
very similar to the expression that has been  investigated in~\cite{PomShp}.
Thus several results and methods of~\cite{PomShp} apply directly to 
this expression too.

The density of integers in~\cite[Corollary~3]{KS} can be evaluated via
the classical
Wirsing theorem~\cite{Wirs} and a result of  Chinen and  Murata~\cite{ChMu}.

More precisely, let $\ell_p$ denote the multiplicative order of 2 modulo 
an odd prime $p$. Kl{\o}ve and
Schwartz~\cite[Corollary~3]{KS} show  that for integers $q \equiv 2 \bmod 4$ 
for which  $\ell_p \equiv 0 \bmod 4$ for every odd prime divisor $p\mod q$, 
we have $\omega_{2,1}(q) = (3q+2)/8$ and in fact there is an explicit 
construction that achieves this value, 
see also~\cite[Corollary~3]{KLY}. Note that $\ell_p \equiv 0 \bmod 4$ 
implies that $p \equiv 1 \pmod 4$, since we always have 
$\ell_p \mid p-1$. So in fact we have $q \equiv 2 \bmod 8$ 
and thus $(3q+2)/8 \in \Z$. 

Thus, it is interesting to investigate the number $N(Q)$ of 
such integers $q \le Q$.
We note that the calculations of Kl{\o}ve and Schwartz~\cite[Example~3]{KS} show that 
$N(40002)=1745$.

We now present an asymptotic formula for $N(Q)$. 

We say that a function $f(n)$ defined on positive integers is 
{\it multiplicative\/} for $f(uv) = f(u)f(v)$ for any relatively 
prime integer $u,v\ge 1$.

We recall the classical theorem of
Wirsing~\cite{Wirs}.

\begin{lemma}
\label{lem:Wirs} Assume that a real-valued multiplicative function
$f(n)$ satisfies
the following conditions:
\begin{itemize}
\item $f(n) \ge 0$, $n =1,2, \ldots$;
\item $f(p^\alpha) \le  a b^\alpha$, $\alpha = 2, 3, \ldots$,  for some
constants $a,b> 0$ with $b < 2$;
\item there exists a constant $\tau > 0$ such that
$$
\sum_{p \le x} f(p)  = \(\tau + o(1) \) \frac{x}{\log x}.
$$
\end{itemize}
Then, for $x \to \infty$ we have
$$
\sum_{n \le x} f(n) = \(\frac{1} {e^{\gamma\tau}\Gamma(\tau)} +o(1)\)
\frac{x}{\log x}\prod_{p\leq
x}\,  \sum_{\alpha=0}^\infty \frac{f(p^\alpha)}{p^\alpha},
$$
where $\gamma$ is the Euler constant, and
$$ \Gamma(s) =  \int_0^{\infty} e^{-t} t^{s-1} \; dt  $$
is the $\Gamma$-function.
\end{lemma}

Let $Q_4(x)$ denote the number of odd primes $p\le x$ with 
$\ell_p \equiv 0 \bmod 4$. 
Then by~\cite[Theorem~1.1]{ChMu} we obtain the following result.

\begin{lemma}
\label{lem:order} We have 
$$
Q_4(x)= \(\frac{1}{3} + o(1) \) \frac{x}{\log x}
$$
as $Q \to \infty$. 
\end{lemma}


%

Now we establish an analogue of the Mertens formula.

\begin{lemma}
\label{lem:Ord Mert}  There exists an absolute constant
$\eta$ such that
$$\prod_{\substack{3 \le p\le x\\ \ell_p  \equiv 0 \bmod 4}}  \(1 + \frac{1}{p-1}\)  = \eta
(\log x)^{1/3} + O\((\log x)^{-2/3}\).
$$
\end{lemma}

\begin{proof} 
In view of the fact that
$$
\log\(1 + \frac{1}{p-1}\)=\frac1p+O\(\frac1{p^2}\)
$$
it is equivalent to prove that there exists an absolute constant $\kappa$ such that
\begin{equation}\label{eq: prod sum}
\sum_{\substack{3 \le p\le x\\ \ell_p  \equiv 0 \bmod 4}}\frac1p
= \frac{1}{3}\log\log x+\kappa+ O\(\frac{1}{\log x}\).
\end{equation}
Let us define the function $\thetaord(x)$ by the identity
$$
\thetaord(x)=\sum_{\substack{3 \le p\le x\\ \ell_p  \equiv 0 \bmod 4}}
\frac{\log p}p.
$$
Observe that by Lemma~\ref{lem:order}, we have
\begin{equation*}
\begin{split}
\thetaord(x)&=
\frac{Q_4(x)  \log x}{x} + \int_2^x \frac{\log t -
1}{t^2} Q_4(t)\,dt \\
& = \frac{1}{3} \int_2^x \frac{\log t - 1}{t^2} \pi(t)\,dt  + O\(1\),
\end{split}
\end{equation*}
where, as usual, $\pi(t)$ denotes the number of primes $p \le t$.

The same arguments also imply that
$$
\sum_{ p\le x } \frac{\log p}{p} = \int_2^x\frac{\log t  - 1
}{t^2}\pi(t)\,dt + O\( 1\)
$$
and  the  Mertens theorem, see~\cite[Chapter~1, Theorem~3.1]{Prach}, yields
$$\thetaord(x)=\sum_{\substack{3 \le p\le x\\ \ell_p  \equiv 0 \bmod 4}}\frac{\log p}p=
\frac{1}{3} \log x+R(x)
$$
for some function $R(x)$ with $R(x)=O(1)$. We now derive
\begin{equation*}
\begin{split}
& \sum_{\substack{3 \le p\le x\\ \ell_p  \equiv 0 \bmod 4}}\frac1p 
= \frac{\thetaord(x)}{\log x}
+\int_2^x\frac{\thetaord(t)}{t (\log t)^2} \,dt \\
&\qquad = \frac{1}{\log x}\(\frac{1}{3}  \log x + R(x)\) + 
\int_2^x\frac{1}{t (\log t)^2}\(\frac{1}{3} \log t + R(t)\) \,dt \\
&\qquad =  \frac{1}{3}  \log \log x  - \frac{1}{3}  \log \log 2 + \frac{1}{3}  +
\int_2^x\frac{
R(t)}{t (\log t)^2} \,dt+  O\(\frac{1}{\log x}\)\\
&\qquad =  \frac{1}{3}  \log \log x  - \frac{1}{3}  \log \log 2 + \frac{1}{3}  +
\int_2^\infty \frac{ R(t)}{t (\log t)^2} \,dt+ O\(\frac{1}{\log x}\)
\end{split}
\end{equation*}
(here the existence of the improper integral follows 
from $R(t)= O(1)$).
So we now obtain~\eqref{eq: prod sum} with
$$
\kappa =    \frac{1-\log \log 2}{3}
+ \int_2^\infty \frac{ R(t)}{t (\log t)^2} \,dt, 
$$
which concludes the proof. 
\end{proof}

We are now ready to establish an asymptotic 
formula for $N(Q)$.
 
\begin{theorem}
\label{thm:dens} There is an absolute constant $\rho> 0$ such that 
We have
$$
N(Q) =  (\rho + o(1)) \frac{Q}{(\log Q)^{2/3}}
$$
as $Q \to \infty$. 
\end{theorem}

\begin{proof} Let us define the multiplicative function $f(n)$ 
by its values on prime powers $p^\alpha$, $\alpha =1,2, \ldots$, 
$$
f(p^\alpha) = 
\left\{  \begin{array}{ll}
1,&  \ \text{if}\ p\ge 3\ \text{and}\ \ell_p  \equiv 0 \bmod 4; \\
0, & \  \text{otherwise}.
\end{array} \right.
$$
Then 
$$
N(Q) = \sum_{n\le (Q-2)/2}f(n).
$$
Applying Lemma~\ref{lem:Wirs}, where by Lemma~\ref{lem:order}
we can take $\tau = 1/3$, we derive
\begin{equation}
\label{eq:NQ}
N(Q) =  \(\frac{1} {e^{\gamma/3}\Gamma(1/3)} +o(1)\)
\frac{Q}{2\log Q} \prod_{p\le (Q-2)/2} \sum_{\alpha=0}^\infty \frac{f(p^\alpha)}{p^\alpha}.
\end{equation}

We note that 
$$
\sum_{\alpha=0}^\infty \frac{f(p^\alpha)}{p^\alpha} = 
\left\{  \begin{array}{ll}
1 + 1/(p-1) ,& \  \text{if}\ p\ge 3\ \text{and}\ \ell_p  \equiv 0 \bmod 4; \\
1, &  \ \text{otherwise}.
\end{array} \right.
$$
Hence,  using Lemma~\ref{lem:Ord Mert}, we conclude the proof. 
\end{proof}

It is certainly interesting to get a closed form expression 
for the constant $\rho$ in Theorem~\ref{thm:dens} or 
at least evaluate it numerically.

\section{Remarks}

 Kl{\o}ve and
Schwartz~\cite[Theorem~3]{KS} showed that if $g$ is  a primitive
root modulo a prime $p$ and the interval $\cM=\{ -\mu,-\mu+1,\ldots, \lambda\}$
 contains $\delta$ consecutive powers of $g$,
then
$$
\omega_{\lambda,\mu}(p)\le  \rf{\frac{p-1}{\delta}} +1.
$$
Unfortunately, one expects that $\delta$ is rather small if, say
$\lambda+\mu <p/2$ and thus this approach does not seem to be able
to produce strong results. For example, by a result of
Bourgain~\cite[Theorem~B]{Bour4}, for a primitive root $g$ modulo a
prime $p$, the sequence of fractional parts
$$
\left\{g^n/p\right\}, \qquad n=1, \ldots, \delta,
$$
is uniformly distributed modulo $1$,  provided that
$\delta > p^{C/\log \log p}$
for some absolute constant $C$. Thus, for any fixed $\varepsilon> 0$
and a sufficiently large $p$, we have
$$
\delta \le p^{C/\log \log p}
$$
for any $\lambda$ and $\mu$ with $\lambda+\mu < (1-\varepsilon)p$.
Using~\cite{Bour1,Bour2,Bour3} one can obtain similar results for arbitrary composite
moduli.

Furthermore, if $\lambda+\mu$ is small (but not very small),
namely if
$$
p^\varepsilon \le \max\{\lambda, \mu\} < (\sqrt{1/2}-\varepsilon)p^{1/2},
$$
then using a different result of Bourgain~\cite[Theorem~1]{Bour5} we can get
$$
\delta = o((\log p)^{\psi(p)})
$$
where $\psi(p)$ is an arbitrary function with $\psi(p)\to \infty$ as $p \to \infty$.
Indeed, if say $\lambda > 0$ then by~\cite[Theorem~1]{Bour5} the set
$$
\left\{m g^n/p\right\}, \qquad m =1, \ldots, \lambda, \ n=1, \ldots, \delta,
$$
is uniformly distributed modulo $1$. On the other hand if $g^n \in \cM$, $n=1, \ldots, \delta$,
then these elements are  all at the distance at least $1/2 - (\sqrt{1/2}-\varepsilon)^2 =
(\sqrt{2} - \varepsilon)\varepsilon$ from $1/2$.

We call a set $\cS$ of size $ \# \cS=N$ with
$$
\nu_{\lambda,\mu}(q,N)=\#(\cM \cS)= \#\cM \# \cS=(\mu+\lambda)N
$$ 
a $(\lambda,\mu;q)$-\emph{packing set} of order $N$.
In~\cite{KBE,KLNY,KLY}, the authors applied \emph{packing sets} to
define codes that correct single limited-magnitude errors. It is
certainly interesting to find constructions of such sets and in
particular obtain non-trivial estimates on the introduced quantity
in~\cite{KS}
$$
\vartheta_{\lambda,\mu}(q)=\max\{N~:~\nu_{\lambda,\mu}(q,N)=(\mu+\lambda)N\}.
$$
We note that it is very easy to achieve an asymptotically optimal value of
$\nu_{\lambda,\mu}(q,N)$.
Indeed, we may restrict ourselves to the case $\lambda\ge \mu$.
If $N\lambda < q$, we simply take  $\cS=\{1,\ldots,N\}$ and using the
classical asymptotic formula for the
average value of the square of the divisor function, see~\eqref{eq:sum tau2}, 
and  the Cauchy-Schwarz inequality, we obtain
$$
\(\#\cM \# \cS\)^2 \le  \#(\cM \cS) \sum_{k\le N\lambda} \tau^2(k) \ll  \#(\cM \cS)
\#\cM \# \cS \(\log q\)^3.
$$
Hence, $\nu_{\lambda,\mu}(q,N) \ge  c(\lambda +\mu) N  (\log q)^{-3}$ for
some absolute constant $c>0$. In fact the result of Koukoulopoulos~\cite{Kou}
yields an even tighter bound. 

If $N\lambda \ge q$, we consider only the subset $\{1,\ldots  \fl{q/N}\}$ of $\cM$ and get similarly the bound 
$\nu_{\lambda,\mu}(q)  \ge c  q (\log q)^{-3}$. However, investigating
when $\nu_{\lambda,\mu}(q,N)=(\mu+\lambda)N$ and thus estimating
$\vartheta_{\lambda,\mu}(q)$ seems to be more challenging.

\section*{Acknowledgements}

Parts of this paper were written during a very pleasant visit of the first
 author to RICAM, Austrian Academy of Sciences in Linz. He wishes to thank for the hospitality.

During the preparation of this work,
Z.X.C. was partially supported by the National Natural Science
Foundation of China grant 61373140
 and the Special Scientific Research Program
in Fujian Province Universities of China  under grant  JK2013044;
I.S. was partially supported by Australian Research Council grant DP130100237
and by Singapore National Research Foundation grant
CRP2-2007-03.

\end{document}